\documentclass{amc}
\usepackage{amsmath}
\usepackage[all]{xy}

\usepackage[colorlinks=true]{hyperref}
\hypersetup{urlcolor=blue, citecolor=red}

\def\currentvolume{xx}
 \def\currentissue{x}
  \def\currentyear{20xx}
   
    \def\ppages{xxx-xxx}
     \def\DOI{xxx}

\usepackage{amsmath,amssymb}  
\usepackage{cite}
     
\usepackage{mathtools}

\usepackage{arydshln}

\newtheorem{theorem}{Theorem}[section]

\theoremstyle{definition}
\newtheorem{definition}[theorem]{Definition}
\newtheorem{remark}[theorem]{Remark}

\title[On shortened and punctured cyclic codes]
      {On shortened and punctured cyclic codes}

\author[Arti Yardi and Ruud Pellikaan]{}


\thanks{$^*$ Corresponding author}

\begin{document}
\maketitle

\centerline{\scshape Arti Yardi$^*$}
\medskip
{\footnotesize
  \centerline{IRIT/INP-ENSEEIHT, University of Toulouse}
   \centerline{Toulouse, France}
}

\bigskip

\centerline{\scshape Ruud Pellikaan}
\medskip
{\footnotesize
  \centerline{Department of Mathematics and Computer Science}
   \centerline{Eindhoven University of Technology}
    \centerline{P.O. Box 513, 5600 MB Eindhoven, The Netherlands}
}

\bigskip
\centerline{(Communicated by xxx)}

\begin{abstract}
The problem of identifying whether the family of cyclic codes is asymptotically good or not is a long-standing
open problem in the field of coding theory. It is known in the literature that some families of cyclic codes such
as BCH codes and Reed-Solomon codes are asymptotically bad, however in general the answer to this question is not
known. A recent result by Nelson and Van Zwam shows that, all linear codes can be obtained by a sequence of
puncturing and/or shortening of a collection of asymptotically good codes~\cite{Nelson_2015}. In this paper, we prove that any linear
code can be obtained by a sequence of puncturing and/or shortening of some cyclic code. Therefore the result that all
codes can be obtained by shortening and/or puncturing cyclic codes leaves the possibility open that cyclic codes are
asymptotically good.
\end{abstract}

%
\section{Introduction}
%
%
A family of codes is called as \textit{asymptotically good} if it contains 
an infinite sequence of codes such that both the rate and the ratio of minimum distance to length of every code in this sequence
is bounded away from zero~\cite[Ch.~9]{Macwilliams_Sloane_1977}.
The problem of identifying asymptotically good families of codes have been well studied in the 
literature~\cite{Justesen_1972,Justesen_1973, Lin_67, Berlekamp_74,Roth_1992,Schulman_2002,Ling_2003}.
Justesen codes, AG codes, and Expander codes are some examples of the known families of asymptotically good codes.
While the family of BCH codes is known to be asymptotically bad~\cite{Lin_67},  
some families of quasi-cyclic codes have been identified as asymptotically good~\cite{Ling_2003}. 
However, whether the family of cyclic codes is asymptotically good or not is 
a long-standing open problem in the literature~\cite{Assmus1965cyclic, P_Charpin_open}.

%
This problem of identifying whether cyclic codes are asymptotically good or not was first addressed by 
Assmus, Mattson, and Turyn~\cite{Assmus1965cyclic}. Since then this problem has been studied by
various researchers~\cite{Lin_67,Berlekamp_74,Berman_1967, Massey_repeated_cyclic_1991,Perez_2006}. 
Lin and Weldon proved that the family of BCH codes is asymptotically bad~\cite{Lin_67}. 
Berlekamp and Justesen constructed a family of cyclic codes that performs better than BCH codes, 
however this family also turned out to be asymptotically bad~\cite{Berlekamp_74}.
Berman provided a necessary condition for the sequence of cyclic codes to be asymptotically good~\cite{Berman_1967}. 
He proved that the necessary condition for the sequence of cyclic codes to be asymptotically good is that  
the number of distinct prime factors of the lengths of cyclic codes should tend to infinity.
Mart{\'\i}nez-P{\'e}rez and Willems have also provided 
similar necessary conditions for cyclic codes to be asymptotically good and provided some classes 
of cyclic codes that were also identified as asymptotically bad~\cite{Perez_2006}.
Castagnoli et al.~showed that if there exists an asymptotically good sequence of cyclic codes (simple-root or repeated-root), then 
an asymptotically good sequence of simple-root cyclic codes can also be identified~\cite{Massey_repeated_cyclic_1991}.

%
While the existing literature provides a sequence of asymptotically bad cyclic codes, 
the possibility of existence of a good sequence of cyclic codes cannot be denied.
A recent result by Nelson and Van Zwam shows that all linear codes can be obtained by a sequence of puncturing 
and/or shortening of a collection of asymptotically good codes~\cite{Nelson_2015}. 
This result provides a necessary condition for a given class of codes to be asymptotically good.
One can use this result to give an independent proof of the fact that graphic and co-graphic codes are not asymptotically 
good as shown by Kashyap~\cite{N_Kashyap_2008}.
In this paper, we prove that the class of cyclic codes satisfy this necessary condition, i.e., 
any linear code can be obtained by a sequence of puncturing and/or shortening of some cyclic code.
Therefore the result that all codes can be obtained by shortening and puncturing cyclic codes leaves the possibility 
open that cyclic codes are asymptotically good.

%
The remaining paper is organized as follows.
In Section~\ref{Section_Preliminaries}, we provide some notation and preliminaries required in the paper. 
The main result of the paper is provided in Section~\ref{Section_main_result}, followed by
some concluding remarks in Section~\ref{Section_Conclusion}.

%
\section{Notation and preliminaries}
\label{Section_Preliminaries}
The finite field with $q$ elements is denoted by $\mathbb{F}_q$ 
and the polynomial ring with coefficients from $\mathbb{F}_q$ is denoted by $\mathbb{F}_q[X]$.
Without loss of generality, we consider any vector as a row vector and
use boldface letters to indicate vectors.
The components of a vector are indicated by lowercase letters.
For example, a vector $\mathbf{v} \in \mathbb{F}_q^n$ is given by 
$\mathbf{v} = \begin{bmatrix} v_0 & v_1 & \ldots & v_{n-1} \end{bmatrix}$,
where $v_i \in \mathbb{F}_q$ is the $i$th component of $\mathbf{v}$, for $i = 0,1,\ldots,n-1$.
The polynomial representation of $\mathbf{v}$ is given by $\mathbf{v}(X)=v_0+v_1X+v_2X^2+\ldots+v_{n-1}X^{n-1}$.
An all-zero vector of length $n$ is denoted by $\mathbf{0}_n$.
The collection of linear block codes of length $n$ and dimension $k$ is denoted by $C(n,k)$ and 
a linear block code in this collection in denoted by $C$.
The cyclic code of length $n$ and generator polynomial $g(X) \in \mathbb{F}_q[X]$ is denoted by 
$C(n,g)$.
Throughout this paper, we consider linear block codes that are defined over $\mathbb{F}_q$.

We shall next recall the definition of puncturing and shortening operations on a linear block code $C$.
\begin{definition}
\label{Definition_Puncturing}
Puncturing of a code~\cite[Ch.~1]{Huffman_Pless_ECC}:
Let $C$ be an $[n, k]$ linear block code over $\mathbb{F}_q$ and let $\mathcal{L}$ be the set of any $l$ coordinate locations. 
Then the puncturing operation on $C$ at coordinate locations in the set $\mathcal{L}$
consists of deleting the entries of every codeword in $C$ at locations in the set $\mathcal{L}$.
\end{definition}
\begin{definition}
\label{Definition_Shortening}
Shortening of a code~\cite[Ch.~1]{Huffman_Pless_ECC}:
Let $C$ and $\mathcal{L}$ be as defined in Definition~\ref{Definition_Puncturing}. 
Then the shortening operation on $C$ at coordinate locations in the set $\mathcal{L}$
consists of two steps. In the first step, consider the set $\mathcal{W}$ of codewords in $C$ that 
have zeros at the locations in the set $\mathcal{L}$. In the second step, puncturing
operation is performed on $\mathcal{W}$ at coordinate locations in the set $\mathcal{L}$.
\end{definition}

\begin{remark}
It is known that the code obtained after the above mentioned puncturing and shortening operation is a linear block code of length $n-l$~\cite[Ch.~1]{Huffman_Pless_ECC}.
\end{remark}

%
\section{Main result}
\label{Section_main_result}
In this section, we provide the main result of the paper in the following theorem.

\begin{theorem}
\label{Theorem_main_LBC_by_shortening_puncturing}
Any $\mathbb{F}_q$ linear block code can be obtained by a sequence of puncturing and/or shortening of some cyclic code.
\end{theorem}
\begin{proof}
Let $C$ be any $\mathbb{F}_q$ linear block code of length $n$ and dimension $k$.
Let $\{ \mathbf{v}_1, \mathbf{v}_2, \ldots, \mathbf{v}_k \}$ be a basis of $C$.
Using this basis, we now construct a cyclic code such that by a sequence of
puncturing and/or shortening of this cyclic code, it is possible to obtain code $C$.
Corresponding to $\{ \mathbf{v}_1, \mathbf{v}_2, \ldots, \mathbf{v}_k \}$, 
define a vector $\mathbf{f}$ as follows
\begin{equation}
\begin{aligned}
\mathbf{f} &\coloneqq \Big[ 1 \mbox{~~} [\mathbf{v}_1 \mbox{~~} \mathbf{0}_n]  \mbox{~~} [\mathbf{v}_2 \mbox{~} \mathbf{0}_n] \mbox{~~} [\mathbf{v}_3 \mbox{~} \mathbf{0}_n] \mbox{~~} \cdots \mbox{~~} [\mathbf{v}_{k-1} \mbox{~} \mathbf{0}_n] \mbox{~~}  \mbox{~} \mathbf{v}_k \mbox{~~} 1 \Big],\\
	   &= \hspace{0.035in} \Big[f_0 \mbox{~~} f_1 \mbox{~~} f_2 \mbox{~~} \ldots \mbox{~~} f_{m}\Big], 
\end{aligned}
\label{Eqn_f}
\end{equation}
where $m = 2n(k-1) + n + 1$ and $f_0, f_1, \ldots, f_m$ are the components of $\mathbf{f}$ such that $f_0=f_m=1$.
Let $f(X)$ be the polynomial corresponding to $\mathbf{f}$.
Let $p$ be the characteristic of the finite field $\mathbb{F}_q$.
Considering the coefficients of $f(X)$, define a vector $\mathbf{g}$ as follows
\begin{align}
\mathbf{g} &\coloneqq \Big[ f_0 \mbox{~~} \mathbf{0}_{p-1} \mbox{~~} f_1  \mbox{~~} \mathbf{0}_{p-1} \mbox{~~} \cdots \mbox{~~} \mathbf{0}_{p-1} \mbox{~~} f_{m-1}\mbox{~~} f_m \Big],
\end{align}
where $\mathbf{0}_{p-1}$ is an all-zero vector of length $p-1$.
Let $g(X)$ be the polynomial corresponding to $\mathbf{g}$. It is given by,
\begin{align}
g(X) \coloneqq \sum_{i=0}^{m-1} f_i X^{pi} + X^{[(m-1)p + 1]}.
\label{Eqn_gX}
\end{align}
The derivative $g^{\prime}(X)$ of $g(X)$ is given by,
\begin{align}
g^{\prime}(X) &= \sum_{i=0}^{m-1} pf_i X^{pi-1} + \big[(m-1)p + 1\big] X^{(m-1)p} \\
	      &= X^{(m-1)p}, 
\end{align}
where the last equality is obtained since $p=0$ in $\mathbb{F}_q$.
Since $g(0)$ is not equal to zero, it follows that the greatest common divisor of $g(X)$ and $g^{\prime}(X)$ is equal to one. 
This implies that $g(X)$ does not have multiple zeros~\cite[Ch.~1, Thm.~1.86]{Lidl86}.
Furthermore there exists an extension $\mathbf{F}_{q^e}$ of $\mathbb{F}_q$ such that all zeros of $g(X)$ are non-zero elements of this finite extension, 
where $e$ is some positive integer.
Hence $g(X)$ divides $X^N-1$ where $N=q^e-1$ and $\mbox{gcd}(N,q)=1$.
Let $n^{\prime}$ be the shortest integer such that $g(X)$ divides $X^{n^{\prime}}-1$, $\mbox{gcd}(n^{\prime},q)=1$,
and $2(n^{\prime}-\deg(g)) \geq n^{\prime}$. 
Note that if $g(X)$ has multiple zeros, it will never divide $X^{n^{\prime}}-1$ for some integer $n^{\prime}$ 
such that $\mbox{gcd}(n^{\prime},q)=1$~\cite{Lidl86}.
The necessity of the condition $2(n^{\prime}-\deg(g)) \geq n^{\prime}$ will be explained later in the proof. 
From (\ref{Eqn_f}) we have $f_m=1$ and this implies that $g(X)$ is a monic polynomial.
It is known that $g(X)$ that satisfies the above mentioned conditions generates the cyclic code $C(n^{\prime},g)$
of length $n^{\prime}$~\cite[Ch.~4]{Huffman_Pless_ECC}.
Let $k^{\prime} = n^{\prime}-\deg(g)$ be the dimension of $C(n^{\prime},g)$. 
The condition $2(n^{\prime}-\deg(g)) \geq n^{\prime}$ implies that 
$2(n^{\prime}-\deg(g)) = 2 k^{\prime} \geq n^{\prime}$, i.e., 
the rate of the code $C(n^{\prime},g)$ is greater than or equal to $1/2$.
We next prove that $C(n^{\prime},g)$ is the required cyclic code of the proof.

Let $G$ be a generator matrix of $C(n^{\prime},g)$.
In order to perform a sequence of puncturing and/or shortening operations on $C(n^{\prime},g)$, 
we will first write the rows of $G$ in a convenient form. 
The first row $\mathbf{g}_0$ of $G$ can be given by,
\begin{align}
\mathbf{g}_0 &\coloneqq \Big[ f_0 \mbox{~~} \mathbf{0}_{p-1} \mbox{~~} f_1  \mbox{~~} \mathbf{0}_{p-1} \mbox{~~} \cdots \mbox{~~} \mathbf{0}_{p-1} \mbox{~~} f_{m-1} \mbox{~~} f_m \mbox{~~} \mathbf{0}_{k^{\prime}-1} \Big].
\label{Eqn_G_first_row_pre}
\end{align}
The $i$th row $\mathbf{g}_i$ of $G$ can be obtained by considering $i$ right cyclic shifts of $\mathbf{g}_0$~\cite{LinCostello2004},
i.e., $\mathbf{g}_i$ is given by,
\begin{align}
\mathbf{g}_i \coloneqq \Big[  \mathbf{0}_{i} \mbox{~~} f_0 \mbox{~~} \mathbf{0}_{p-1} \mbox{~~} f_1  \mbox{~~} \mathbf{0}_{p-1} \mbox{~~} \cdots \mbox{~~} \mathbf{0}_{p-1} \mbox{~~} f_{m-1} \mbox{~~} f_m \mbox{~~} \mathbf{0}_{k^{\prime}-1-i} \Big],
\label{Eqn_G_ith_row_pre}
\end{align}
where $i = 0, 1,\ldots, k^{\prime}-1$.
The condition $2k^{\prime} \geq n^{\prime}$ implies that $k^{\prime} \geq n^{\prime} - k^{\prime} = \deg(g)$. 
Recall that $\deg(g) = (m-1)p+1$ and this implies that $k^{\prime}-1 \geq (m-1)p$.
Thus the condition $2k^{\prime} \geq n^{\prime}$ ensures that $G$ has at least $(m-1)p$ rows.
The generator matrix $G$ can now be written as
\footnotesize
\begin{align}
%
\begin{bmatrix}
 \mathbf{g}_0\\
 \vdots\\
 \mathbf{g}_{p}\\
 \vdots\\ 
 \mathbf{g}_{(m-1)p}\\
 \vdots \\
 \mathbf{g}_{k^{\prime}-1} 
\end{bmatrix} 
=
\left[
\hspace{-0.05in}
\begin{array}{c c c c c :c:  c c c c c  c c c c c c c}
f_0 & \mathbf{0}_{p-1} & f_1 & \mathbf{0}_{p-1}  & \cdot   & f_{m-1} & f_m & \mathbf{0}_{p-2} & 0    & \cdot & \cdot & \cdot & 0\\
\vdots & &  &  \vdots       &   &  \vdots  &\vdots &  \vdots    &  \vdots   & &  &   & \vdots   \\
0 & \mathbf{0}_{p-1} & f_0 & \mathbf{0}_{p-1}  & \cdot  & f_{m-2} & 0 & \mathbf{0}_{p-2} &  f_{m-1} & f_m & 0   & \cdot  & 0\\
\vdots &  &  \vdots       &    &    & \vdots&  &  \vdots   & &  \ddots &  & &\vdots  \\
0 & \cdot & \cdot & \cdot   &  \cdot & f_0 & 0 & \mathbf{0}_{p-2} & f_1 & \cdot  & f_m  & \cdot  & 0\\
\vdots &  &  \vdots       &    &      & &  \vdots & & &   & &  \ddots & \vdots   \\
0 & \cdot & \cdot & \cdot   & \cdot   & \cdot & \cdot & \cdot &  & \cdot & \cdot  & \cdot   & f_m\\
\end{array} 
\hspace{-0.05in}
\right], 
\label{Eqn_G_matrix_pre}
\end{align}
\normalsize
where the coordinate location of the coefficient $f_{m-1}$ in row $\mathbf{g}_0$ is same as that of 
the coordinate location of the coefficient $f_{m-i-1}$ in row $\mathbf{g}_{ip}$, for $i = 1,2,\ldots, m-1$, 
which is indicated using the two vertical dashed lines.
The code $C$ can now be obtained by a sequence of puncturing and/or shortening of $C(n^{\prime},g)$ in the following five steps.

~\\
A] \textit{Shortening of $C(n^{\prime},g)$ at coordinate locations in the range $\big[ (j-1)p + 2, jp\big]$, 
for $j = 1,2,\ldots, m-1$}
 
We first choose $j=1$ and shorten $C(n^{\prime},g)$ in the range $[2, p]$. 
Let $\mathcal{W}_1$ be the vector space spanned by the rows $\mathbf{g}_0$ and $\mathbf{g}_{p}, \mathbf{g}_{p+1}, \ldots, \mathbf{g}_{k^{\prime}-1}$, 
i.e., the first row of $G$ and all the rows after $\mathbf{g}_{p}$ of $G$ in (\ref{Eqn_G_matrix_pre}).
The vectors in $\mathcal{W}_1$ have zeros in the coordinate location range $[2,p]$.
We now prove that the set of codewords in $C(n^{\prime},g)$ that have zeros in this range is exactly 
the vector space $\mathcal{W}_1$.
Consider the matrix formed by the rows $\mathbf{g}_0, \mathbf{g}_1, \ldots, \mathbf{g}_{p}$ of $G$, i.e., 
the initial $p+1$ rows of $G$ as follows,
\begin{align}
\left[
\begin{array}{c}
 \mathbf{g}_0\\ 
 \mathbf{g}_1\\
 \vdots\\  
 \mathbf{g}_{p-1}\\ 
 \mathbf{g}_{p}\\
\end{array}
\right]
=
\left[
\begin{array}{c :c  c  c  c c  c c : c c c  c c c c c c c c c c}
f_0  & 0 & \cdot & \cdot & \cdot & \cdot & \cdot & 0 & f_1 &  0 & \cdot & \cdot &  \cdot &  \cdot &  \cdot & 0\\
0 & f_0 & 0 & \cdot & \cdot & \cdot & \cdot & 0  & 0 & f_1 &  \cdot & \cdot &  \cdot &  \cdot &  \cdot & 0\\
\vdots &  &  \ddots  &  &         &     &    &\vdots     &   & &    \ddots   \\
0 & \cdot & \cdot & \cdot & \cdot & \cdot & 0 & f_0 & 0  & \cdot & \cdot &  \cdot & \cdot &  \cdot &  \cdot & 0\\
0 & \cdot & \cdot & \cdot & \cdot & \cdot & \cdot & 0 & f_0 & \cdot  & \cdot & \cdot &  \cdot & \cdot &  \cdot & 0\\
\end{array} 
\right],
\label{Eqn_G1_matrix_intermediate}
\end{align}
where the coordinate location range between the two vertical dashed lines in (\ref{Eqn_G1_matrix_intermediate}) is $[2,p]$.
Recall that $f_0=1$ (see (\ref{Eqn_f})).
From (\ref{Eqn_G_matrix_pre}) and (\ref{Eqn_G1_matrix_intermediate}) this implies that a non-zero codeword in $C(n^{\prime},g)$ that
is a linear combination of the rows $\mathbf{g}_{1}, \mathbf{g}_{2}, \ldots, \mathbf{g}_{p-1}$, 
will have at least one non-zero coordinate in the range $[2,p]$. 
Thus the set of codewords in $C(n^{\prime},g)$ that have zeros in the range $[2,p]$ 
is exactly the vector space $\mathcal{W}_1$.

We will now shorten $C(n^{\prime},g)$ in the range $[2,p]$. 
In the first step of shortening, we will get subspace $\mathcal{W}_1$ (see Definition~\ref{Definition_Shortening}). 
In the second step, puncturing operation will be performed
in the coordinate location range $[2,p]$ of $\mathcal{W}_1$.
A generator matrix $G_1^{\prime}$ of the code obtained by shortening $C(n^{\prime},g)$ in this range is given by,
\begin{align}
G_1^{\prime}
=
\left[
\hspace{-0.05in}
\begin{array}{c c c c c :c:  c c c c c  c c c c c c c}
f_0  & f_1 & \mathbf{0}_{p-1}  & \cdot & \cdot  & f_{m-1} & f_m & \mathbf{0}_{p-2} & 0   & \cdot & \cdot & \cdot & \cdot  & \cdot& 0\\
0  & f_0   & \mathbf{0}_{p-1} & \cdot & \cdot  & f_{m-2} & 0 & \mathbf{0}_{p-2} &  f_{m-1} & f_m & 0 & \cdot & \cdot & \cdot  & 0\\
\vdots   &         & \ddots   &    & &\vdots  &  \vdots   & &   &\ddots & & & &&\vdots  \\
0  & \cdot & \cdot   & \cdot &  \cdot & f_0 & 0 & \mathbf{0}_{p-2} & f_1 & \cdot & \cdot  & f_m & 0 & \cdot  & 0\\
\vdots   &         &    &     & & &  \vdots & & & &  &\ddots &  & &\vdots    \\
0  & \cdot & \cdot   & \cdot  & \cdot & \cdot & \cdot & \cdot &  & \cdot & \cdot  & \cdot & \cdot & \cdot  & f_m\\
\end{array} 
\hspace{-0.05in}
\right].
\label{Eqn_G1_matrix_intermediate2}
\end{align}

Observe that shortening of $C(n^{\prime},g)$ in the coordinate location range $[2,p]$ eliminated the rows $\mathbf{g}_1, \mathbf{g}_2, \ldots, \mathbf{g}_{p-1}$
of $G$. Thus the number of rows of $G_1^{\prime}$ is equal to $k^{\prime}-(p-1)$.
Using similar arguments it can be proved that shortening of $C(n^{\prime},g)$ in the range $\big[ (j-1)p + 2, jp\big]$ 
will eliminate the rows $\mathbf{g}_{(j-1)p + 1}, \mathbf{g}_{(j-1)p+2}, \ldots, \mathbf{g}_{jp-1}$ of $G$
for $j = 1,2,\ldots, m-1$. A generator matrix $G_1$ of the code obtained by shortening in these sets is given by,
\footnotesize
\begin{align}
G_1
=
\left[
\hspace{-0.05in}
\begin{array}{c c c c c c : c c c c c  c c c c c c c}
f_0  & f_1 & f_2  & \cdot & \cdot  & f_{m-1} & f_m & \mathbf{0}_{p-2} & 0   & \cdot  & \cdot & \cdot & \cdot & 0\\
0  & f_0   & f_1 & \cdot & \cdot  & f_{m-2} & 0 & \mathbf{0}_{p-2} &  f_{m-1} & f_m  & \cdot & \cdot & \cdot  & 0\\
\vdots   &  \vdots       &    &    & &\vdots  &  \vdots   & &   &\ddots & & &\vdots  \\
0  & \cdot & \cdot   & \cdot &  \cdot & f_0 & 0 & \mathbf{0}_{p-2} & f_1 & \cdot & \cdot  & f_m  & \cdot  & 0\\
\vdots   &  \vdots       &    &     & & &  \vdots & & &   &\ddots &      \\
0  & \cdot & \cdot   & \cdot  & \cdot & \cdot & \cdot & \cdot &  & \cdot & \cdot  & \cdot  & \cdot & f_m\\
\end{array} 
\hspace{-0.05in}
\right]
\eqqcolon
\begin{bmatrix}
 \mathbf{h}_0\\ 
 \mathbf{h}_1\\
 \vdots\\  
 \mathbf{h}_{m-1}\\ 
 \vdots\\
 \mathbf{h}_{k_1} 
\end{bmatrix},
\label{Eqn_G1_matrix}
\end{align}
\normalsize
where $k_1 = \big[k^{\prime} -1\big] - \big[(m-1)(p-1)\big]$ and $\mathbf{h}_0, \mathbf{h}_1, \ldots, \mathbf{h}_{k_1}$ are 
defined as the rows of $G_1$.
Let $C_1$ be the linear block code generated by $G_1$ defined in (\ref{Eqn_G1_matrix}).

~\\
B] \textit{Shortening of $C_1$ at the last $k^{\prime}-1 - [2np(k-1)]$ locations} 

From (\ref{Eqn_G1_matrix}) we have $k_1 \geq m-1$. 
Substituting the value of $m = 2n(k-1) + n + 1$ we get $k_1 \geq 2n(k-1) + n$, which
implies that $k_1 > 2n(k-1)$.
Thus the condition $k_1 \geq m-1$ ensures that $G_1$ has at least $2n(k-1)$ rows.
In order to perform this shortening operation, we first write the generator matrix $G_1$ 
in a convenient form.
We substitute the value of the vector 
$[f_0 \mbox{~} f_1 \mbox{~} \ldots \mbox{~} f_{m-1}] = [1 \mbox{~} \mathbf{v}_{1} \mbox{~} \mathbf{0}_{n} \mbox{~} \mathbf{v}_{2} \mbox{~} \mathbf{0}_{n} \mbox{~} \cdots  \mbox{~} \mathbf{0}_{n} \mbox{~}  \mathbf{v}_{k}]$ 
from (\ref{Eqn_f}) and
form a matrix using the rows $\mathbf{h}_0$ and $\mathbf{h}_{2n}$ of $G_1$ as follows,
\begin{align}
\left[
\begin{array}{c}
 \mathbf{h}_0\\
 \mathbf{h}_{2n}\\
\end{array}
\right]
& =
\left[
\begin{array}{c c  c c :c :c  c c c c c  c c c c c c c }
[1 \mbox{~} \mathbf{v}_{1} \mbox{~} \mathbf{0}_{n}] & \mathbf{v}_{2} & \cdot  & \cdot & \mathbf{v}_{k} & f_m & 0 &  \cdot &  \cdot &  \cdot &    \cdot & \cdot & 0\\
 \mbox{[} \mathbf{0}_{2n} \mbox{~~} 1 ]   & \mathbf{v}_{1} & \cdot & \cdot  & \mathbf{v}_{k-1} & 0 & \cdot & \cdot  & f_m & 0 &  \cdot & \cdot & 0 \\
\end{array} 
\right].
\label{Eqn_G_g0_g2n}
\end{align}
From (\ref{Eqn_G_g0_g2n}), it can be seen that the coordinate location of the vector $\mathbf{v}_k$ in $\mathbf{h}_0$
and the location of the vector $\mathbf{v}_{k-1}$ in $\mathbf{h}_{2n}$ is the same, which is indicated using the two vertical dashed lines. 
In general, the coordinate location of $\mathbf{v}_{k-j}$ in $\mathbf{h}_{2nj}$ will be same as that of 
the location of $\mathbf{v}_k$ in $\mathbf{h}_0$, for $j = 1,2, \ldots, k-1$.
Using this the generator matrix $G_1$ can be written as
\footnotesize
\begin{align}
%
\begin{bmatrix}
 \mathbf{h}_0\\
 \vdots\\
 \mathbf{h}_{2n}\\
 \vdots\\ 
 \mathbf{h}_{2n(k-1)}\\
 \mathbf{h}_{2n(k-1)+1}\\
 \vdots\\
 \mathbf{h}_{k_1}
\end{bmatrix}
=
\left[
\hspace{-0.05in}
\begin{array}{c c  c c c c  c c c c :c  c c c c c c c }
[1 \mbox{~} \mathbf{v}_{1} \mbox{~} \mathbf{0}_{n}] & \mathbf{v}_{2} & \cdot  & \cdot & \mathbf{v}_{k} & f_m & 0   &  \cdot & & \cdot & \cdot & \cdot   & \cdot & 0\\
\vdots & \vdots &    &  &       \vdots  &     &   & \ddots &   & &  &&          &    \vdots   \\
 \mbox{[} \mathbf{0}_{2n} \mbox{~} 1 ]   & \mathbf{v}_{1} & \cdot & \cdot  & \mathbf{v}_{k-1} & \cdot & \cdot  & f_m   & & \cdot & \cdot & \cdot & \cdot  & 0 \\
\vdots & \vdots &    &  &       \vdots  &     &     & \ddots & & &  &&          &    \vdots   \\
\mathbf{0}_{2n+1}  & \mathbf{0}_n & \cdot & \cdot & \mathbf{v}_{1}  & \cdot & \cdot & \cdot & \cdot    & f_m & 0  & \cdot & \cdot & 0\\
\mathbf{0}_{2n+1}  & \mathbf{0}_n & \cdot & \cdot & \cdot & \cdot & \cdot & \cdot & \cdot    &   & f_m & 0 & \cdot  & 0\\
\vdots & \vdots &    &  &       \vdots  &     &     & \vdots & & &  &       & \ddots  &    \vdots   \\
\mathbf{0}_{2n+1}  & \mathbf{0}_n & \cdot & \cdot & \cdot & \cdot & \cdot  & \cdot   &  \cdot &  & \cdot & \cdot & \cdot  & f_m\\
\end{array} 
\hspace{-0.05in}
\right].
\label{Eqn_G1_matrix_good_form}
\end{align}
\normalsize

Note that the row $\mathbf{h}_{2n(k-1)}$ of $G_1$ corresponds to the row $\mathbf{g}_{2np(k-1)}$ of $G$. 
From (\ref{Eqn_G_ith_row_pre}) this implies that the number of columns on the right hand side 
the vertical dashed line in (\ref{Eqn_G1_matrix_good_form}) is equal to $k^{\prime}-1 - [2np(k-1)]$.
Using similar arguments as in step A], it can be shown that the shortening of $C_1$
in the last $k^{\prime}-1 - [2np(k-1)]$ locations will provide a code with generator
matrix $G_2$ given by
\begin{align}
G_2=
\left[
\begin{array}{c c  c c c :c  c c c c c  c c  }
[1 \mbox{~} \mathbf{v}_{1} \mbox{~} \mathbf{0}_{n}] & \mathbf{v}_{2} & \cdot  & \cdot & \mathbf{v}_{k} & f_m & 0 & 0 & 0 &  \cdot & \cdot & \cdot & 0\\
\vdots & \vdots &    &  &       \vdots  &     &   & \ddots &  &  &        & &   \vdots   \\
 \mbox{[} \mathbf{0}_{2n} \mbox{~} 1 ]   & \mathbf{v}_{1} & \cdot & \cdot  & \mathbf{v}_{k-1} & \cdot & \cdot & \cdot & f_m  & 0 & \cdot & \cdot  & 0 \\
\vdots & \vdots &    &  &       \vdots  &     &   &  & \ddots &  &        & &   \vdots   \\
\mathbf{0}_{2n+1}  & \mathbf{0}_n & \cdot & \cdot & \mathbf{v}_{1}  & \cdot & \cdot & \cdot & \cdot  &  \cdot & \cdot & \cdot & f_m 
\end{array} 
\right].
\label{Eqn_G2_matrix}
\end{align}
Let $C_2$ be the linear block code generated by $G_2$ defined in (\ref{Eqn_G2_matrix}).

~\\
C] \textit{Puncturing of $C_2$ at the locations other than the initial $m$ columns} 

Comparing (\ref{Eqn_G1_matrix}) and (\ref{Eqn_G2_matrix}), the number of columns on the left hand side of the vertical dashed
line in (\ref{Eqn_G2_matrix}) is equal to $m$. Thus the generator matrix $G_3$ of the code obtained
after this puncturing operation is given by,
\begin{align}
G_3=
\left[
\begin{array}{c c  c c c c  c c c c c  c c  }
[1 \mbox{~} \mathbf{v}_{1} \mbox{~}\mathbf{0}_{n}] & [\mathbf{v}_{2} \mbox{~} \mathbf{0}_{n}] & \cdot  & \cdot & \mathbf{v}_{k} \\
\vdots & \vdots &    &       &            \vdots   \\
 \mbox{[} \mathbf{0}_{2n} \mbox{~} 1 ]   & [\mathbf{v}_{1} \mbox{~} \mathbf{0}_{n}] & \cdot & \cdot  & \mathbf{v}_{k-1}  \\
\vdots & \vdots &    &      &              \vdots   \\
\mathbf{0}_{2n+1}  & \mathbf{0}_{2n} & \cdot & \cdot & \mathbf{v}_{1}  
\end{array} 
\right]
\eqqcolon
\left[
\begin{array}{c}
 \mathbf{b}_0\\
 \vdots\\
 \mathbf{b}_{2n}\\
 \vdots\\ 
 \mathbf{b}_{2n(k-1)}\\
\end{array}
\right],
\label{Eqn_G3_matrix_pre}
\end{align}
where $\mathbf{b}_0, \mathbf{b}_1, \ldots, \mathbf{b}_{2n(k-1)}$ are defined as the rows of $G_3$.
Suppose the codeword $\mathbf{v}_j$ is given by 
\begin{align}
\mathbf{v}_j &= \begin{bmatrix} \mathbf{0}_{l_j} & \mathbf{w}_j & \mathbf{0}_{r_j} \end{bmatrix} 
\label{Eqn_v_division}
\end{align}
where $\mathbf{0}_{l_j}$ and $\mathbf{0}_{r_j}$ are all-zero vectors of lengths $l_j$ and $r_j$
respectively and $\mathbf{w}_j$ is the vector obtained by puncturing the initial $l_j$ and the last $r_j$ entries of $\mathbf{v}_j$
such that the first and the last entry of $\mathbf{w}_j$ are not zero, for $j = 1,2,\ldots,k$.
Let $d_1, d_2, \ldots, d_k$ be the lengths of $\mathbf{w}_1, \mathbf{w}_2, \ldots, \mathbf{w}_k$ respectively, i.e., 
$l_j+d_j+r_j=n$. 
Without loss of generality we assume that $r_1 \leq r_2 \leq \ldots \leq r_k$.
Substituting the value of $\mathbf{v}_j$, the generator matrix $G_3$ can be written as
\begin{align}
%
\begin{bmatrix}
 \mathbf{b}_{0}\\
 \vdots \\
 \mathbf{b}_{2n}\\
 \mathbf{b}_{2n+1}\\
 \vdots\\ 
 \mathbf{b}_{2n(k-1)}
\end{bmatrix} 
=
\left[
\begin{array}{c c  c  c c c:  c: c c c c  c c c c }
1 & \mathbf{0}_{l_1} & \mathbf{w}_{1} & \mathbf{0}_{r_1} & \mathbf{0}_{n-1} & 0 & \big[\mathbf{0}_{l_2} \mbox{~~} \mathbf{w}_{2} \mbox{~~} \mathbf{0}_{r_2} \big] & \cdot & \cdot  & \cdot & \mathbf{v}_{k} \\
\vdots &  &  \vdots  &  &         &   \vdots  &   & &        &    & \vdots   \\
0 & \mathbf{0}_{l_1} & \mathbf{0}_{d_1} & \mathbf{0}_{r_1} & \mathbf{0}_{n-1} & 1 & \big[ \mathbf{0}_{l_1} \mbox{~~} \mathbf{w}_{1} \mbox{~~} \mathbf{0}_{r_1} \big] & \cdot & \cdot  & \cdot & \mathbf{v}_{k-1} \\
0 & \mathbf{0}_{l_1} & \mathbf{0}_{d_1} & \mathbf{0}_{r_1}& \mathbf{0}_{n-1} & 0 & \big[ 1 \mbox{~~} \mathbf{0}_{l_1} \mbox{~~} \mathbf{w}_{1} \mbox{~~} \mathbf{0}_{r_1-1} \big] & \cdot & \cdot & \cdot & \cdot \\
\vdots &  &  \vdots  &  & &           \vdots &         \vdots  & & & & \vdots \\
0 & \mathbf{0}_{l_1} & \mathbf{0}_{d_1}  & \mathbf{0}_{r_1} & \mathbf{0}_{n-1} & 0 & \mathbf{0}_{n} & \cdot  & \cdot & \cdot & \mathbf{v}_{1} 
\end{array} 
\right],
\label{Eqn_G3_matrix}
\end{align}
where the number of coefficients between the two vertical dashed lies in (\ref{Eqn_G3_matrix}) is equal to $n$.
Let $C_3$ be the linear block code generated by $G_3$.

~\\
D] \textit{Shortening of the code $C_3$ at coordinate locations in the range $\big[(2j-1)n+1-r_1,2jn \big]$, for $j=1,2, \ldots, k-1$ }

We first choose $j=1$ and shorten $C_3$ in the range $\big[n+1-r_1, 2n\big]$.
Let $\mathcal{W}_2$ be the vector space spanned by the rows $\mathbf{b}_0$ and $\mathbf{b}_{2n}, \mathbf{b}_{2n+1}, \ldots, \mathbf{b}_{2n(k-1)}$, 
of $G_3$, i.e., the first row of $G_3$ and all the rows after $\mathbf{b}_{2n}$ in (\ref{Eqn_G3_matrix}).
A generator matrix of $\mathcal{W}_2$ is given by
\begin{align}
%
\begin{bmatrix}
 \mathbf{b}_{0}\\
 \mathbf{b}_{2n}\\
 \mathbf{b}_{2n+1}\\
 \vdots\\ 
 \mathbf{b}_{2n(k-1)}
\end{bmatrix} 
=
\left[
\begin{array}{c c  c : c c: c  c c c c c  c c c c }
1 & \mathbf{0}_{l_1} & \mathbf{w}_{1} & \mathbf{0}_{r_1} & \mathbf{0}_{n-1} & 0 & \big[\mathbf{0}_{l_2} \mbox{~~} \mathbf{w}_{2} \mbox{~~} \mathbf{0}_{r_2} \big] & \cdot & \cdot  & \cdot & \mathbf{v}_{k} \\
0 & \mathbf{0}_{l_1} & \mathbf{0}_{d_1} & \mathbf{0}_{r_1} & \mathbf{0}_{n-1} & 1 & \big[ \mathbf{0}_{l_1} \mbox{~~} \mathbf{w}_{1} \mbox{~~} \mathbf{0}_{r_1} \big] & \cdot & \cdot  & \cdot & \mathbf{v}_{k-1} \\
0 & \mathbf{0}_{l_1} & \mathbf{0}_{d_1} & \mathbf{0}_{r_1}& \mathbf{0}_{n-1} & 0 & \big[ 1 \mbox{~~} \mathbf{0}_{l_1} \mbox{~~} \mathbf{w}_{1} \mbox{~~} \mathbf{0}_{r_1-1} \big] & \cdot & \cdot & \cdot & \cdot \\
\vdots &  &  \vdots  &  & &           \vdots &         \vdots  & & & & \vdots \\
0 & \mathbf{0}_{l_1} & \mathbf{0}_{d_1}  & \mathbf{0}_{r_1} & \mathbf{0}_{n-1} & 0 & \mathbf{0}_{n} & \cdot  & \cdot & \cdot & \mathbf{v}_{1} 
\end{array} 
\right],
\label{Eqn_G4_W2}
\end{align}
where the range $\big[n+1-r_1, 2n\big]$ is indicated by the two vertical dashed lines. It can be seen that every vector in $\mathcal{W}_2$
has zeros in this range. We next prove that the set of codewords in $C_3$ that has zeros in this range is exactly the vector space $\mathcal{W}_2$.

Suppose $[ \mathbf{0}_{l_1} \mbox{~} \mathbf{w}_1] = [w_{1} \mbox{~~} w_{2} \mbox{~~} \ldots \mbox{~} w_{d}]$, where $d=l_1+d_1$.
Using this the matrix formed by the rows $\mathbf{b}_0, \mathbf{b}_1, \ldots, \mathbf{b}_{2n}$ of $G_3$, i.e., the initial $2n+1$ rows of $G_3$ is
as follows,
\footnotesize
\begin{align}
%
\begin{bmatrix}
 \mathbf{b}_0\\ 
 \mathbf{b}_1\\
 \vdots\\  
 \mathbf{b}_{d}\\ 
 \mathbf{b}_{d+1}\\
 \vdots\\
 \mathbf{b}_{2n-1}\\ 
 \mathbf{b}_{2n}
\end{bmatrix} 
=
\left[
\hspace{-0.05in}
\begin{array}{c c  c  c  :c c  c c c c c  c :c c c c c c c c c}
1 & w_1   & \cdot & w_{d} & 0 & \cdot & \cdot &  \cdot &  \cdot &  \cdot & \cdot & 0 & 0 &  [\mathbf{0}_{l_2} \mbox{~} \mathbf{w}_2 \mbox{~} \mathbf{0}_{r_2} ] & \cdot   &  \mathbf{v}_k\\
0 & 1   & \cdot & \cdot & w_{d} & 0 & \cdot &  \cdot &  \cdot &  \cdot & \cdot & 0& 0 &  \cdot \mbox{~~~} \cdot & \cdot   &  \cdot\\
 \vdots       &  &       &  &     &   \ddots &     &   & &  & & & &  & &    \vdots   \\
0   & \cdot & \cdot & 1 & w_1 & \cdot & \cdot &  w_d &  0 &  \cdot & \cdot & 0 & 0 &  \cdot \mbox{~~~} \cdot & \cdot   &  \cdot\\
0   & \cdot & \cdot & \cdot & 1 & w_1 & \cdot & \cdot &  w_d &  0  & \cdot & 0 & 0 &  \cdot \mbox{~~~} \cdot & \cdot   &  \cdot\\
 \vdots       &  &       &  &     &    &     & \ddots  & &  & & & &  & &    \vdots   \\
0   & \cdot & \cdot & \cdot & \cdot & \cdot & \cdot &  \cdot &  \cdot &  \cdot & \cdot & 1 & \cdot &  \cdot \mbox{~~~} \cdot & \cdot   &  \cdot \\
0   & \cdot & \cdot & \cdot & \cdot & \cdot & \cdot &  \cdot &  \cdot &  \cdot & \cdot & 0 & 1 &  [\mathbf{0}_{l_1} \mbox{~} \mathbf{w}_1 \mbox{~} \mathbf{0}_{r_1} ] & \cdot   &  \mathbf{v}_{k-1}
\end{array} 
\hspace{-0.05in}
\right].
\label{Eqn_G4_initial_2n}
\end{align}
\normalsize
where the range $\big[n+1-r_1, 2n\big]$ is indicated by the two vertical dashed lines.
Since $w_d$ is not equal to zero (see (\ref{Eqn_v_division})), 
it can be seen that a non-zero codeword in $C_3$ that is a linear combination
of the rows $\mathbf{b}_{1}, \mathbf{b}_{2}, \ldots, \mathbf{b}_{d}$, 
will have at least one non-zero entry in the range $\big[d+2, 2d+2 \big] \subseteq \big[n+1-r_1, 2n\big]$. 
Further, a non-zero codeword in $C_3$ that is a linear combination 
of the rows $\mathbf{b}_{d+1}, \mathbf{b}_{d+2}, \ldots, \mathbf{b}_{2n-1}$, 
will have at least one non-zero entry in the range $\big[n+1-r_1, 2n \big]$.
Therefore the set of codewords in $C_3$ that have zeros in the range $\big[n+1-r_1, 2n\big]$
is exactly the vector space $\mathcal{W}_2$. 

Using steps similar to step A], it can be shown that shortening of $C_3$ in the range $\big[n+1-r_1, 2n\big]$
will eliminate the rows $\mathbf{b}_1, \mathbf{b}_2, \ldots, \mathbf{b}_{2n-1}$ of $G_3$, i.e., 
when $j=1$, the rows of $G_3$ between $\mathbf{b}_0$ to $\mathbf{b}_{2n}$ were eliminated. 
We next prove that shortening of $C_3$ in the range $\big[(2j-1)n+1-r_1, 2jn\big]$, 
the rows of $G_3$ between $\mathbf{b}_{2n(j-1)}$ to $\mathbf{b}_{2jn}$ will get eliminated, for $j = 1,2,\ldots, k-1$.
The generator matrix $G_3$ with a focus on these rows is given by
\footnotesize
\begin{align}
\begin{bmatrix}
\mathbf{b}_0\\ 
\vdots\\  
\mathbf{b}_{2n}\\
\vdots\\  
\mathbf{b}_{2n(j-2)}\\ 
\vdots\\  
\mathbf{b}_{2n(j-1)}\\
\mathbf{b}_{2n(j-1)+1}\\
 \vdots\\
 \mathbf{b}_{2nj}\\ 
 \vdots \\
 \mathbf{b}_{2n(k-1)}
\end{bmatrix}
=
\left[
\hspace{-0.05in}
\begin{array}{c c : c  :c  c :c  :c c c c c  c c c c c c c c c c}
1   &  \cdot &  [\mathbf{0}_{l_j} \mbox{~~~} \mathbf{w}_j \mbox{~~~} \mathbf{0}_{r_j} ] & \mathbf{0}_{n-1} & 0  &  [\mathbf{0}_{l_{j+1}} \mbox{~} \mathbf{w}_{j+1} \mbox{~} \mathbf{0}_{r_{j+1}} ] & \cdot & \cdot &  \mathbf{v}_k\\
\vdots         &      & \vdots   &     &    & \vdots    &  & & \vdots   \\
0   &  \cdot &  [\mathbf{0}_{l_{j-1}} \mbox{~} \mathbf{w}_{j-1} \mbox{~} \mathbf{0}_{r_{j-1}} ] & \mathbf{0}_{n-1} & 0   &  [\mathbf{0}_{l_j} \mbox{~~~} \mathbf{w}_j \mbox{~~~} \mathbf{0}_{r_j} ] & \cdot & \cdot &  \mathbf{v}_{k-1}\\
\vdots         &      & \vdots   &     &    & \vdots    &  & & \vdots   \\
0   &  \cdot &  [\mathbf{0}_{l_{2}} \mbox{~~~} \mathbf{w}_{2} \mbox{~~~} \mathbf{0}_{r_{2}} ] & \mathbf{0}_{n-1} & 0  &  [\mathbf{0}_{l_3} \mbox{~~~} \mathbf{w}_3 \mbox{~~~} \mathbf{0}_{r_3} ] & \cdot & \cdot &  \cdot\\
\vdots         &      & \vdots   &     &    & \vdots    &  & & \vdots   \\
0   &  \cdot &  [\mathbf{0}_{l_{1}} \mbox{~~~} \mathbf{w}_{1} \mbox{~~~} \mathbf{0}_{r_{1}} ] & \mathbf{0}_{n-1} & 0  &  [\mathbf{0}_{l_2} \mbox{~~~} \mathbf{w}_2 \mbox{~~~} \mathbf{0}_{r_2} ] & \cdot & \cdot &  \cdot\\
0   &  \cdot &  [1 \mbox{~~} \mathbf{0}_{l_{1}} \mbox{~~} \mathbf{w}_{1} \mbox{~~} \mathbf{0}_{r_{1}-1} ] & \mathbf{0}_{n-1} & 0   &  [0 \mbox{~~} \mathbf{0}_{l_2} \mbox{~~} \mathbf{w}_2 \mbox{~~} \mathbf{0}_{r_2-1} ] & \cdot & \cdot &  \cdot\\
\vdots         &      & \vdots   &     &    & \vdots    &  & & \vdots   \\
0   &  \cdot &  \cdot \mbox{~~~} \cdot & \mathbf{0}_{n-1} & 1  &  [\mathbf{0}_{l_1} \mbox{~~~} \mathbf{w}_1 \mbox{~~~} \mathbf{0}_{r_1} ] & \cdot & \cdot &  \mathbf{v}_{k-j}\\
\vdots         &      & \vdots   &     &    & \vdots    &  & & \vdots   \\
0   &  \cdot &  \cdot \mbox{~~~} \cdot &  \cdot &  \cdot  &   \cdot & \cdot & \cdot &  \mathbf{v}_{1}
\end{array} 
\hspace{-0.05in}
\right].
\label{Eqn_G4_induction}
\end{align}
\normalsize
Observe that the structure of the rows $\mathbf{b}_{2n(j-1)}, \mathbf{b}_{2n(j-1)+1}, \ldots, \mathbf{b}_{2nj}$ of the matrix in (\ref{Eqn_G4_induction})
is the same as that of the rows $\mathbf{b}_0, \mathbf{b}_1, \ldots, \mathbf{b}_{2n}$ the matrix $G_3$ in (\ref{Eqn_G3_matrix}).
Further, since $r_1 \leq r_2 \leq \ldots \leq r_k$, it can be seen that the vector space spanned by the 
rows $\mathbf{b}_0, \mathbf{b}_{2n}, \ldots, \mathbf{b}_{2n(j-1)}$ 
and the all the rows after row $\mathbf{b}_{2n(j-1)}$ have zeros in the range $\big[(2j-1)n+1-r_1,2jn \big]$. 
Using similar arguments as that of the case when $j=1$ it can be proved that shortening in this range will eliminate
the rows of $G_3$ between $\mathbf{b}_{2n(j-1)}$ to $\mathbf{b}_{2jn}$.
A generator matrix $G_4$ of the code obtained after the above mentioned set of shortening
operations is given by
\begin{align}
G_4=
\left[
\begin{array}{c}
 \mathbf{b}_{0}\\
 \mathbf{b}_{2n}\\
 \vdots\\ 
 \mathbf{b}_{2n(k-1)}\\
\end{array}
\right]
=
\left[
\begin{array}{c c  c  c c c  c c:c c c  c c c c }
1 & \mathbf{0}_{l_1} & \mathbf{w}_{1} &  0 & \big[\mathbf{0}_{l_2} \mbox{~~} \mathbf{w}_{2} \mbox{~~} \mathbf{0}_{r_2-r_1} \big] & \cdot & \cdot & \cdot & \mathbf{v}_{k} \\
0 & \mathbf{0}_{l_1} & \mathbf{0}_{d_1} & 1 & \big[ \mathbf{0}_{l_1} \mbox{~~~} \mathbf{w}_{1} \big]& \cdot & \cdot & \cdot & \mathbf{v}_{k-1} \\
\vdots &  &  \vdots  &  &           \vdots &     &   &  \vdots   \\
0 & \cdot & \cdot  & \cdot & \cdot & \cdot & \cdot & \cdot & \mathbf{v}_{1} 
\end{array} 
\right].
\label{Eqn_G4_matrix}
\end{align}
Let $C_4$ be the code corresponding to the generator matrix $G_4$ in (\ref{Eqn_G4_matrix}).

~\\
E] \textit{Puncturing of $C_4$ at the coordinate locations other than the last $n$ locations}

The puncturing operation in this range will delete the columns of $G_4$ that are on the left hand side of the vertical dashed line in (\ref{Eqn_G4_matrix}).
A generator matrix $G_5$ of the code obtained by this puncturing operation is given by,
\begin{align}
G_3=
\left[
\begin{array}{c}
\mathbf{v}_{k}  \\
\mathbf{v}_{k-1}  \\
\vdots      \\
\mathbf{v}_{1}  
\end{array} 
\right].
\label{Eqn_G5_matrix}
\end{align}

Observe that, $G_5$ is a generator matrix of the required linear block code $C$ of the theorem and this completes the proof.
\end{proof}


\section{Conclusion}
\label{Section_Conclusion}

In this paper, we proved that any linear block code can be obtained by a sequence of puncturing and/or shortening
of some cyclic code. While the result is in itself interesting, due to the recent result given by Nelson and Vam
Zwam~\cite{Nelson_2015}, our result may have applications in studying the long-standing open problem of deciding whether the
family of cyclic codes is asymptotically good or not. The result given by Nelson and Vam Zwam says that, given
a family of asymptotically good codes, any linear block code can be obtained by a sequence of puncturing and/or
shortening of some code in this family. Our result essentially proves that the family of cyclic codes satisfies this
condition.

\section*{Acknowledgments}
The first author is supported by ANR-11-LABEX-0040-CIMI within the program ANR-11-IDEX-0002-02 of
the Centre International de Math{\'e}matiques et Informatique de Toulouse, France.
The first author would also like to acknowledge the support of the Bharti Centre for Communication at IIT Bombay, India.

%

\medskip

Received for publication xxx.

\medskip

{\it E-mail address: }artidilip.yardi@irit.fr\\
\indent{\it E-mail address: }g.r.pellikaan@tue.nl

\end{document}